\documentclass[prl,twocolumn,twoside,superscriptaddress]{revtex4}

\usepackage{graphicx,epic,eepic,epsfig,subfigure,amsmath,latexsym,amssymb,verbatim,color,bbm}

\def\squareforqed{\hbox{\rlap{$\sqcap$}$\sqcup$}}
\def\qed{\ifmmode\squareforqed\else{\unskip\nobreak\hfil
\penalty50\hskip1em\null\nobreak\hfil\squareforqed
\parfillskip=0pt\finalhyphendemerits=0\endgraf}\fi}
\def\endenv{\ifmmode\;\else{\unskip\nobreak\hfil
\penalty50\hskip1em\null\nobreak\hfil\;
\parfillskip=0pt\finalhyphendemerits=0\endgraf}\fi}

\newtheorem{theorem}{Theorem}

\newtheorem{corollary}[theorem]{Corollary}

\newtheorem{definition}[theorem]{Definition}

\newtheorem{lemma}[theorem]{Lemma}

\newenvironment{proof}[1][Proof]{\noindent\textbf{Proof.} }{\hfill\qed}
\newenvironment{proof-of}[1][Proof]{\noindent\textbf{Proof~#1.} }{\hfill\qed}

\newcommand{\nc}{\newcommand}
\nc{\rnc}{\renewcommand}
\nc{\bra}[1]{\langle#1|}
\nc{\ket}[1]{|#1\rangle}
\nc{\ketbra}[2]{|#1\rangle\!\langle #2|}
\nc{\braket}[2]{\langle #1 | #2  \rangle}
\nc{\proj}[1]{| #1\rangle\!\langle #1 |}
\nc{\avg}[1]{\langle#1\rangle}
\nc{\sfrac}[2]{\mbox{$\frac{#1}{#2}$}}
\nc{\ox}{\otimes}
\nc{\dg}{\dagger}
\nc{\lbar}[1]{\overline{#1}}
\nc{\rar}{\rightarrow}
\nc{\dn}{\downarrow}
\nc{\lrar}{\longrightarrow}
\nc{\tr}{\operatorname{Tr}}
\nc{\var}{\operatorname{var}}
\nc{\Rank}{\operatorname{Rank}}
\nc{\polylog}{\operatorname{polylog}}
\nc{\id}{{\operatorname{id}}}
\nc{\1}{{\openone}}
\nc{\di}{\mathrm{d}}

\nc{\cA}{\mathcal{A}}  \nc{\cB}{\mathcal{B}}  \nc{\cC}{\mathcal{C}}
\nc{\cD}{\mathcal{D}}  \nc{\cE}{\mathcal{E}}  \nc{\cF}{\mathcal{F}}
\nc{\cG}{\mathcal{G}}  \nc{\cH}{\mathcal{H}}  \nc{\cI}{\mathcal{I}}
\nc{\cJ}{\mathcal{J}}  \nc{\cK}{\mathcal{K}}  \nc{\cL}{\mathcal{L}}
\nc{\cM}{\mathcal{M}}  \nc{\cN}{\mathcal{N}}  \nc{\cO}{\mathcal{O}}
\nc{\cP}{\mathcal{P}}  \nc{\cQ}{\mathcal{Q}}  \nc{\cS}{\mathcal{S}}
\nc{\cT}{\mathcal{T}}  \nc{\cX}{\mathcal{X}}  \nc{\cZ}{\mathcal{Z}}

\nc{\RR}{{{\mathbb R}}} \nc{\CC}{{{\mathbb C}}} \nc{\FF}{{{\mathbb F}}}
\nc{\NN}{{{\mathbb N}}} \nc{\ZZ}{{{\mathbb Z}}} \nc{\PP}{{{\mathbb P}}}
\nc{\QQ}{{{\mathbb Q}}} \nc{\UU}{{{\mathbb U}}} \nc{\EE}{{{\mathbb E}}}

           \def\sep{\mathinner{\mathrm{SEP}}}
\nc{\LO}{\mathsf{LO}}            \nc{\LOCCONE}{\mathsf{LOCC_1}}	
\nc{\LOCCONEP}{\mathsf{LOCC_1^{\parallel}}}
\nc{\ALL}{{\mathsf{ALL}}}        \nc{\M}{{\mathsf{M}}}
\nc{\NP}{{\mathsf{NP}}}          \nc{\QMA}{{\mathsf{QMA}}}
\nc{\SymQMA}{{\mathsf{SymQMA}}}  \nc{\SAT}{{\mathsf{SAT}}}

\begin{document}

\title{Quantum de Finetti theorem under fully-one-way adaptive measurements}
\author{Ke Li}
  \email{carl.ke.lee@gmail.com}
  \affiliation{IBM TJ Watson Research Center, Yorktown Heights, NY 10598, USA}
  \affiliation{Center for Theoretic Physics,  Massachusetts Institute of Technology, Cambridge, MA 02139, USA}
\author{Graeme Smith}
  \email{gsbsmith@gmail.com}
  \affiliation{IBM TJ Watson Research Center, Yorktown Heights, NY 10598, USA}

\date{\today}

\begin{abstract}
  We prove a version of the quantum de Finetti theorem:  permutation-invariant quantum states
  are well approximated as a probabilistic mixture of multi-fold product states. The
  approximation is measured by distinguishability under fully one-way LOCC (local operations
  and classical communication) measurements. Our result strengthens Brand\~{a}o and Harrow's
  de Finetti theorem where a kind of partially one-way LOCC measurements was used for measuring
  the approximation, with essentially the same error bound.  As main applications, we show (i)
  a quasipolynomial-time algorithm which detects multipartite entanglement with amount larger
  than an arbitrarily small constant (measured with a variant of the relative entropy of
  entanglement), and (ii) a proof that in quantum Merlin-Arthur proof systems, polynomially
  many provers are not more powerful than a single prover when the verifier is restricted to
  one-way LOCC operations.
\end{abstract}

\maketitle
Consider random variables $X_1,...,X_n$ representing the color of a sequence of balls drawn
without replacement from a bag of $100$ red balls and $100$ blue balls. These variables are
not independent, since the probability of withdrawing a red ball on the $k$th withdrawl depends
on the number of balls of each color remaining. They are, however, \emph{exchangeable}: the
probability of removing a particular sequence of balls $(x_1,...,x_n)$ is equal to the probability
of removing any reordering of that sequence $(x_{\pi(1)},...,x_{\pi(n)})$ for permuatation $\pi$.
Remarkably, the de Finetti theorem tells us that any such exchangeable random variables can be
represented by independent and identically distributed ones~\cite{deFinetti37, Diaconis-Freedman80},
yeilding a profound result in probability theory and a  powerful tool in statistics.

A series of works have established analogues of this theorem in the quantum domain~\cite{Stormer69,
Hudson-Moody76, Raggio-Werner89, CFS02, Konig-Renner05, CKMR07, Renner07, Brandao-Harrow12}, where
a classical probability distribution is replaced by a quantum state and the situation is more
complicated and interesting, due to entanglement and the existence of many different ways to
distinguish states of multipartite systems. These quantum de Finetti theorems are appealing not
only due to their own elegance on the characterization of symmetric states, but also because of the
successful applications in many-body physics~\cite{Raggio-Werner89, Fannes-Vandenplas06, LNR13},
quantum information~\cite{Renner07, Brandao-Plenio-10, Christandl-Renner12}, and computational
complexity theory~\cite{BSW11, BCY11, Brandao-Harrow12}.

More precisely, a quantum de Finetti theorem concerns the structure of a \emph{symmetric} state
$\rho_{A_1\ldots A_n}$ that is invariant under any permutations over the subsystems~\cite{note0}.
It tells how the reduced state $\rho_{A_1\ldots A_k}$ on a smaller number $k<n$ of subsystems could
be approximated by a mixture of $k$-fold product states, namely, \emph{de Finetti states} of the
form $\int\sigma^{\ox k}\,\di\mu(\sigma)$. Here $\mu$ is a probability measure over density matrices.
Using the conventional distance measure, trace norm, Ref.~\cite{CKMR07} proved a standard de
Finetti theorem with an essentially optimal error bound $2|A|^2k/n$ for the approximation ($|A|$
denotes the dimension of the subsystems). However, in many situations this bound is too large
to be applicable. Luckily it is possible to circumvent
this obstruction. For example, Renner's exponential de Finetti theorem employs the ``almost'' de
Finetti states and has an error bound that decreases exponentially in $n-k$~\cite{Renner07},
being very useful in dealing with cryptography or information theory problems~\cite{Renner07,
Brandao-Plenio-10, Christandl-Renner12}.

In a beautiful work~\cite{Brandao-Harrow12} Brand\~{a}o and Harrow recently proved an LOCC (local
operations and classical communication) de
Finetti theorem, generalizing a similar result for the case $k=2$~\cite{BCY11}.
Both \cite{Brandao-Harrow12} and \cite{BCY11} have overcome the limitation of the standard de Finetti
theorem regarding the dimension dependence. The basic idea is to relax the measure of approximation
by replacing the trace norm with a kind of one-way LOCC norm. This gives an error bound
$\sqrt{\frac{2k^2\ln|A|}{n-k}}$~\cite{note1}, scaling polynomially in $\ln|A|$ instead of polynomially in
$|A|$ as in earlier de Finetti results, which is crucial to the complexity-theoretic applications.

\begin{figure}
\begin{subfigure}[]{}
\includegraphics[width=0.9in]{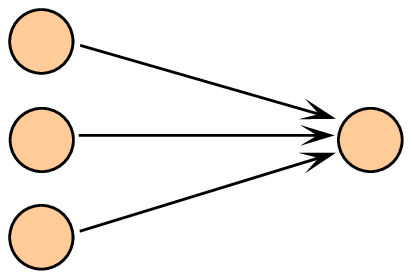}
\end{subfigure}
\hspace{0.35in}
\begin{subfigure}[]{}
\includegraphics[width=1.25in]{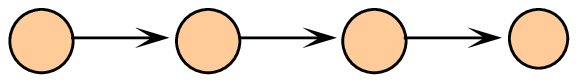}
\end{subfigure}
 \caption{Parallel vs. fully one-way LOCC. (a) $\LOCCONEP$: Parallel one-way LOCC measurements
 used in~\cite{Brandao-Harrow12}. Here the first $k-1$ parties make measurements in parallel and
 report their outcomes to the $k$th, who then makes a measurement that depends on the messages
 he receives. (b) $\LOCCONE$: Fully one-way LOCC measurements. We adopt a more complete
 generalization of one-way LOCC: all the parties measure their own systems sequentially, but in
 a fully adaptive way where each party chooses his own measurement setting depending on the
 outcomes of all the previous measurements performed by the other parties.}
 \label{fig:Parallel-vs-Sequential}
\end{figure}

While \cite{Brandao-Harrow12} showed approximation in the \emph{parallel} one-way LOCC norm
associated with the measurement class $\LOCCONEP$, here we prove a de Finetti theorem where
the approximation is measured with the \emph{fully} one-way LOCC norm (or relative entropy)
associated with $\LOCCONE$ (cf. Fig.~\ref{fig:Parallel-vs-Sequential}).
The error bound remains essentially the same as that of~\cite{Brandao-Harrow12}. This improves
Brand\~{a}o and Harrow's LOCC de Finetti theorem considerably: it is conceptually more complete
and when applied to the problems considered in~\cite{BCY11, Brandao-Christandl11, Brandao-Harrow12}
gives new and improved results. For the problem of entanglement detection, so central to
quantum information theory and experiment, we present strong guarantees for the effectiveness
of the well-known heirarchy of entanglement tests of \cite{DPS}. We also consider the power
of multiple-prover quantum Merlin Arthur games, which bears directly on the problems of
pure-state vs mixed-state $N$-representability~\cite{LCV07} as well as the entanglement properties
of sparse hamiltonian's ground states~\cite{Chailloux-Sattath11}.

\medskip\noindent
{\bf Operational norms as distance measures.}
We identify every positive operator-valued measure $\{M_x\}_x$ with a measurement operation
$\cM$: for any state $\omega$, $\cM(\omega):=\sum_x\proj{x}\tr(\omega M_x)$ with $\{\ket{x}\}_x$
an orthonormal basis. For simplicity we call them both quantum measurement.
Given a class of measurements $\M$, the operational norm is defined as~\cite{MWW09}
\[\| \rho -\sigma \|_\M = \max_{\cM \in \M} \| \cM(\rho)-\cM(\sigma)\|_1.\]
It measures the distinguishability of two quantum states under restricted classes of
measurements. We will be particularly interested in $\| \cdot \|_{\LOCCONE}$ and
$\| \cdot \|_{\LOCCONEP}$. Obviously the former is lower bounded by the latter, since
$\LOCCONEP\subset\LOCCONE$. In fact, these two norms can differ substantially: using a recent
result obtained in~\cite{Aubrun-Lancien14}, we can show for all $d$ there are constant $C$ and
$d\times d\times 2$ states $\rho_{ABC}$ and $\sigma_{ABC}$ such that
$\| \rho_{ABC}-\sigma_{ABC} \|_{\LOCCONE} = 2$ but $\| \rho_{ABC}-\sigma_{ABC} \|_{\LOCCONEP}
\leq C/\sqrt{d}$ (see Appendix).

\medskip\noindent
{\bf Improved LOCC de Finetti theorem.}
Our main result is the following Theorem~\ref{thm:de-Finetti}. Besides the improvement with
the fully one-way LOCC norm, for the first time we employ relative entropy $D(\rho\|\sigma)
=\tr\rho(\log\rho-\log\sigma)$ to measure the approximation, defining $D_\LOCCONE(\rho\|\sigma)
:=\max_{\Lambda\in\LOCCONE}D(\Lambda(\rho)\|\Lambda(\sigma))$.

In the proof, we will use information-theoretic methods similar to \cite{Brandao-Harrow12}, along
with some new ideas. In particular, Lemma~\ref{lemma:multi-to-bi} presented below is a crucial
new technical tool, which may be of independent interest. We employ and manipulate entropic
quantities to derive the final result: apart from relative entropy, the mutual information of a
state $\omega_{AB}$ is defined as $I(A;B):=D(\omega_{AB}\|\omega_A\ox\omega_B)$, and the
conditional mutual information of a state $\omega_{ABC}$ is defined as $I(A;B|C):=I(A;BC)-I(A;C)$.

\begin{theorem}
  \label{thm:de-Finetti}
  Let $\rho_{A_1\ldots A_n}$ be a permutation-invariant state on $\cH_A^{\ox n}$.
  Then for integer $0\leq k\leq n$ there exists a probability measure $\mu$ on density
  matrices on $\cH_A$ such that
  \begin{align}
    \label{eq:de-Finetti-1}
    &D_\LOCCONE\Big(\rho_{A_1\ldots A_k} \big\| \int\sigma^{\ox k}\,\di\mu(\sigma)\Big)\leq
                                                            \frac{(k-1)^2\log |A|}{ n-k}, \\
    \label{eq:de-Finetti-2}
    &\left\|\rho_{A_1\ldots A_k}\!-\!\!\int\!\sigma^{\ox k}\,\di\mu(\sigma)\right\|_{\LOCCONE}
                                                    \!\leq\sqrt{\frac{2(k-1)^2\ln |A|}{ n-k}}.
  \end{align}
\end{theorem}

\begin{proof}
Eq.~(\ref{eq:de-Finetti-2}) follows from Eq.~(\ref{eq:de-Finetti-1}) immediately by using the
Pinsker's inequality~\cite{Fuchs-Graaf99}, $D(\rho\|\sigma)\geq\frac{1}{2\ln2}\|\rho-\sigma\|_1^2$.
So it suffices to prove Eq.~(\ref{eq:de-Finetti-1}).

Group the $n$ subsystems as shown in Fig.~\ref{fig:grouping}: except for one subsystem, the
others are divided into groups of $k-1$ subsystems each (we discard the possibly remaining qubits,
of which there will be fewer than $k-1$). So, we have $m=\lfloor\frac{n-1}{k-1}\rfloor\geq\frac{n-k}{k-1}$
groups. Label the groups as bigger subsystems $B_1, B_2, \ldots, B_m$ and the isolated system
as $A$. Let the $k-1$ subsystems in $B_1$ be $A_1, A_2, \ldots, A_{k-1}$ and the system $A$
is also identified with $A_k$.

\begin{figure}[ht]
  \begin{center}
  \includegraphics[width=7cm]{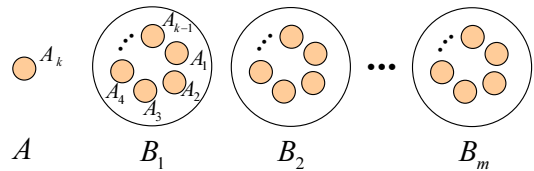}
  \end{center}
  \caption{Grouping and relabeling the $n$ subsystems.}
  \label{fig:grouping}
\end{figure}

Obviously the total state is invariant under permutations over $B_1, B_2, \ldots, B_m$.
So Lemma~\ref{lemma:monogamy} applies. Thus there exists a measurement $\cQ^*: B_2\ldots B_m\rar X$,
such that for any measurement $\cP: B_1\rar Y$ we have
\begin{equation}
  \label{eq:Thproof-1}
  I(A;Y|X)\leq \frac{\log|A|}{m}\leq\frac{(k-1)\log|A|}{n-k}.
\end{equation}
$\cQ^*$ effectively decomposes the state on $AB_1$ into an ensemble. Specifically, we have
$\rho_{AB_1}=\sum_xp_x\rho_{A_1\ldots A_k}^x$, where $p_x$ is the probability of obtaining
the measurement outcome $x$ and $\rho_{A_1\ldots A_k}^x$ is the resulting state on
$A_1\ldots A_k$. Note that since $\rho_{A_1\ldots A_n}$ is permutation-invariant, the
post-measurement states $\rho_{A_1\ldots A_k}^x$ are also permutation-invariant. Now we
rewrite Eq.~(\ref{eq:Thproof-1}) in terms of the relative entropy: for any measurement $\cP$ on
$A_1\ldots A_{k-1}$,
\begin{align}
  &\sum_xp_x D\left(\cP\ox\id^{A_k}(\rho_{A_1\ldots A_k}^x)\big\| \cP
                          (\rho_{A_1\ldots A_{(k-1)}}^x)\ox\rho_{A_k}^x\right) \nonumber \\
  &\leq \frac{(k-1)\log|A|}{n-k}.  \label{eq:Thproof-2}
\end{align}

Pick a one-way LOCC measurement $\Lambda^k$ acting on systems $A_1, \ldots, A_k$ and
denote its reduced measurement on the first $\ell$ systems as $\Lambda^\ell$. Now we
apply Lemma~\ref{lemma:multi-to-bi} to each state $\rho_{A_1\ldots A_k}^x$ and get
\begin{align}
       &D\left(\Lambda^k(\rho_{A_1\ldots A_k}^x) \big\| \Lambda^k(\rho_{A_1}^x
                             \ox\ldots\ox\rho_{A_k}^x)\right)  \label{eq:Thproof-3} \\
  \leq &\sum_{\ell=2}^k D\left(\Lambda^{\ell-1}\ox\id(\rho_{A_1\ldots A_\ell}^x)\big\| \Lambda^{\ell-1}
                             (\rho_{A_1\ldots A_{(\ell-1)}}^x)\ox\rho_{A_\ell}^x\right) \nonumber \\
  \leq &(\!k\!-\!1\!) D\!\left(\!\Lambda^{k-1}\!\ox\!\id(\rho_{A_1\ldots A_k}^x)\big\| \Lambda^{k-1}
                             (\rho_{A_1\ldots A_{(k-1)}}^x)\!\ox\!\rho_{A_k}^x)\!\right), \nonumber
\end{align}
where for the first inequality we have also applied the monotonicity of relative
entropy~\cite{Lindblad-Uhlmann77} and for the second inequality we used the monotonicity of
relative entropy again as well as the symmetry of the state $\rho_{A_1\ldots A_k}^x$.
Combining Eq.~(\ref{eq:Thproof-2}) and Eq.~(\ref{eq:Thproof-3}) we arrive at
\begin{equation}\begin{split}
  \label{eq:Thproof-4}
       &D\Big(\Lambda^k(\rho_{A_1\ldots A_k}) \big\| \Lambda^k(\sum_xp_x \rho_{A_1}^x\ox\ldots\ox\rho_{A_k}^x)\Big) \\
  \leq &\sum_xp_x D\left(\Lambda^k(\rho_{A_1\ldots A_k}^x) \big\| \Lambda^k(\rho_{A_1}^x\ox\ldots\ox\rho_{A_k}^x)\right) \\
  \leq &\frac{(k-1)^2\log|A|}{n-k},
\end{split}\end{equation}
where the first inequality is due to the joint convexity of relative entropy. At this point we
are able to conclude Eq.~(\ref{eq:de-Finetti-1}) from Eq.~(\ref{eq:Thproof-4}), noticing that
$\Lambda^k\in\LOCCONE$ is picked arbitrarily and $\sum_xp_x\rho_{A_1}^x\ox\ldots\ox\rho_{A_k}^x$
is a de Finetti state of the form $\sum_xp_x(\rho_{A}^x)^{\ox k}$ due to the symmetry of
$\rho_{A_1\ldots A_k}^x$.
\end{proof}

\begin{lemma}
  \label{lemma:multi-to-bi}
  Let $\Lambda^k$ be a one-way LOCC measurement on quantum systems $A_1,\ldots, A_k$.
  Denote its reduced measurement corresponding to the first $\ell$ steps on $A_1,\ldots, A_\ell$
  as $\Lambda^\ell$. Then for any state $\rho_{A_1\ldots A_k}$ we have
  \[\begin{split}
     & D\left(\Lambda^k(\rho_{A_1\ldots A_k})\| \Lambda^k(\rho_{A_1}\ox\ldots\ox\rho_{A_k})\right) \\
    =&\sum_{\ell=2}^k D\left(\Lambda^\ell(\rho_{A_1\ldots A_\ell})\| \Lambda^\ell(\rho_{A_1\ldots A_{(\ell-1)}}\ox\rho_{A_\ell})\right).
  \end{split}\]
\end{lemma}

\begin{proof}
It suffices to show
\begin{equation}\begin{split}
  \label{eq:lemm-mtb-proof-1}
   &D\left(\Lambda^k(\rho_{A_1\ldots A_k})\| \Lambda^k(\rho_{A_1}\ox\ldots\ox\rho_{A_k})\right) \\
  =&D\left(\Lambda^{k-1}(\rho_{A_1\ldots A_{k-1}}) \| \Lambda^{k-1}(\rho_{A_1}\ox\ldots\ox\rho_{A_{k-1}})\right) \\
   &+D\left(\Lambda^k(\rho_{A_1\ldots A_k})\| \Lambda^k(\rho_{A_1\ldots A_{k-1}}\ox\rho_{A_k})\right),
\end{split}\end{equation}
because applying this relation recursively allows us to obtain the equation claimed in
Lemma~\ref{lemma:multi-to-bi}.
Write $\Lambda^{k-1}(\rho_{A_1\ldots A_{k-1}})=\sum_xp_x\proj{x}$ and $\Lambda^{k-1}(\rho_{A_1}\ox\ldots\ox\rho_{A_{k-1}})=\sum_xq_x\proj{x}$. Let $\Lambda^k$ be
realized as follows. We first apply $\Lambda^{k-1}$ on $A_1, \ldots, A_{k-1}$. Then depending on
the measurement outcome $x$ we apply a measurement $\cM_x$ on $A_k$. Thus we can write
\begin{align*}
  \Lambda^k(\rho_{A_1\ldots A_k})&=\sum_xp_x\proj{x} \ox \cM_x(\rho_{A_k}^x),   \\
  \Lambda^k(\rho_{A_1\ldots A_{k-1}}\ox\rho_{A_k})&=\sum_xp_x\proj{x} \ox \cM_x(\rho_{A_k}), \\
  \Lambda^k(\rho_{A_1}\ox\ldots\ox\rho_{A_k})&=\sum_xq_x\proj{x} \ox \cM_x(\rho_{A_k}),
\end{align*}
where $\rho_{A_k}^x$ is the state of $A_k$ when $\Lambda^{k-1}$ is applied on $\rho_{A_1\ldots A_k}$
and outcome $x$ is obtained. With these, we can confirm by direct computation that
\begin{equation}\begin{split}
  \label{eq:lemm-mtb-proof-2}
   &D\left(\Lambda^k(\rho_{A_1\ldots A_k})\| \Lambda^k(\rho_{A_1}\ox\ldots\ox\rho_{A_k})\right) \\
  =&D\left(\Lambda^{k-1}(\rho_{A_1\ldots A_{k-1}}) \| \Lambda^{k-1}(\rho_{A_1}\ox\ldots\ox\rho_{A_{k-1}})\right) \\
   &+\sum_xp_x D\left(\cM_x(\rho_{A_k}^x) \|  \cM_x(\rho_{A_k})\right)
\end{split}\end{equation}
and
\begin{equation}\begin{split}
  \label{eq:lemm-mtb-proof-3}
   &D\left(\Lambda^k(\rho_{A_1\ldots A_k})\| \Lambda^k(\rho_{A_1\ldots A_{k-1}}\ox\rho_{A_k})\right)\qquad\quad\; \\
  =&\sum_xp_x D\left(\cM_x(\rho_{A_k}^x) \| \cM_x(\rho_{A_k})\right).
\end{split}\end{equation}
Eq.~(\ref{eq:lemm-mtb-proof-2}) and Eq.~(\ref{eq:lemm-mtb-proof-3}) together lead to
Eq.~(\ref{eq:lemm-mtb-proof-1}) and this concludes the proof.
\end{proof}

\noindent
\emph{Remark.} The quantity $D\left(\rho_{A_1\ldots A_k}\|\rho_{A_1}\ox\ldots\ox\rho_{A_k}\right)$
is sometimes denoted as $I(A_1;A_2;\ldots;A_k)_\rho$ and called the multipartite mutual information.
It is easy to see that $I(A_1;\ldots;A_k)=I(A_1\ldots A_\ell;A_{\ell+1}\ldots
A_k)+I(A_1;\ldots;A_\ell)+I(A_{\ell+1};\ldots;A_k)$. Using this repeatedly we can write the
multipartite mutual information as a sum of bipartite mutual information quantities. This
decomposition can be done in many different ways depending on how we split the subsystems.
Lemma~\ref{lemma:multi-to-bi} is a similar result. However, with the one-way LOCC measurement
$\Lambda^k$, the decomposition only works for our special choice of splitting.

The following lemma, which is a statement of the monogamy of entanglement, is adapted
from~\cite{Brandao-Harrow12}. For completeness we give a proof in the Appendix.

\begin{lemma}
  \label{lemma:monogamy}
  Let $\rho_{AB_1\ldots B_m}$ be a state that is invariant under any permutation over $B_1, B_2,
  \ldots, B_m$. Let $\cP^{B_1\rar Y}$ and $\cQ^{B_2\ldots B_m\rar X}$ be measurement operations
  performed on systems $B_1$ and $B_2\ldots B_m$, respectively. We have
  \[\min_\cQ\max_\cP I(A;Y|X)_{\id^A\ox\cP\ox\cQ(\rho_{AB_1\ldots B_m})}\leq\frac{\log|A|}{m}.\]
\end{lemma}

\medskip\noindent
{\bf Applications.}
By replacing the $\LOCCONEP$ (or Bell) measurements in~\cite{Brandao-Harrow12} with measurments
from $\LOCCONE$, we obtain a couple of interesting results as follows, for which technical proofs
are given in the Appendix.

\medskip\noindent
{\em Detecting multipartite entanglement.}
Deciding whether a density matrix is entangled or separable is one of the most basic problem
in quantum information theory, with both
theoretical and practical significance~\cite{Horodeckis07}. Despite the existence of many
entanglement criteria, up to date the only complete ones that detect all entangled states are
infinite hierarchies~\cite{Horodeckis07}. Among them searching for symmetric extensions is
probably the most useful~\cite{DPS}. This is exactly the scenario where quantum de Finetti
theorems could be expected to be useful.

We consider the situation where a small error $\epsilon$ is permitted, meaning that we must
detect all the entangled states except for those very weak ones that are $\epsilon$-close to
separable (at the same time all the separable states should be detected correctly). This is
equivalently formulated as the Weak Membership Problem for separability: given
a state $\rho_{A_1A_2\ldots A_k}$ that is either separable or $\epsilon$-away from any separable
state, we want to decide which is the case. It has been shown that this problem is $\NP$-hard
when $\epsilon$ is of the order no larger than inverse polynomial of local dimensions (in
trace norm)~\cite{Gurvits03, Gharibian08, Beigi08}. Surprisingly, Brand\~{a}o, Christandl
and Yard found a quasipolynomial-time algorithm for constant $\epsilon$ in one-way
LOCC norm for bipartite states~\cite{BCY11}. This algorithm was generalized to
multipartite states in~\cite{Brandao-Christandl11}, then in~\cite{Brandao-Harrow12} using
a stronger method. These algorithms are all based on the searching for symmetric extensions
of~\cite{DPS}. Along these lines, we present the following result, which is obtained by
applying Theorem~\ref{thm:de-Finetti} to bound the distance between properly extendible
states and separable states.

\begin{corollary}
  \label{cor:det-entanglement}
  Testing multipartite entanglement of a state $\rho_{A_1A_2\ldots A_k}$ with error
  $\epsilon$ can be done via searching for symmetric extensions in time
  \begin{equation}
    \label{eq:det-ent-time}
    \exp\left(c\left(\sum_{i=1}^k\log|A_i|\right)^2 k^2f(\epsilon)\right),
  \end{equation}
  where $f(\epsilon)=\epsilon^{-2}$ if the error is measured by the norm $\|\cdot\|_\LOCCONE$
  and $f(\epsilon)=\epsilon^{-1}$ if it is measured by the relative entropy $D_\LOCCONE$.
\end{corollary}

It is worth mentioning that the run time in Eq.~(\ref{eq:det-ent-time}) is quasipolynomial, 
for constant particle number $k$ and constant error $\epsilon$.
The algorithm in~\cite{Brandao-Christandl11} using $\LOCCONE$-norm behaves exponentially slower
than ours with respect to the number of particles $k$, while the algorithm of~\cite{Brandao-Harrow12}
has the same run-time as ours but works only for $\LOCCONEP$-norm rather than our
$\LOCCONE$-norm approximation. Thus our result has bridged the gap between these two
works. Furthermore, here for the first time we catch the importance
of the \emph{amount of entanglement} in this problem. The quantity $E_r^{\LOCCONE}(\rho):=\min
\{D_\LOCCONE(\rho\|\sigma): \sigma\text{ being separable}\}$, introduced in~\cite{Piani09},
is asymptotically normalized since $E_r^{\LOCCONE}(\Phi_d)=\log (d+1)-1$ for maximally entangled
state $\Phi_d$ of local dimension $d$~\cite{Li-Winter14}. Corollary~\ref{cor:det-entanglement}
shows that, detecting all the $k$-partite entangled states $\rho$ such that $E_r^{\LOCCONE}(\rho)
\geq \epsilon$ can be done in quasi-polynomial time in local dimensions. This is a stronger
statement than using $\LOCCONE$-norm as the error measure.
We point out that for the bipartite case this result can also be obtained by combining the
algorithm of~\cite{BCY11} with the ``commensurate lower bound'' for squashed entanglement
of~\cite{Li-Winter14}.

\medskip\noindent
{\em QMA proof system with multiple proofs.}
$\QMA$, the quantum analogue of the complexity class $\NP$, is the set of decision problems
whose solutions can be efficiently verified on a quantum computer, provided with a
polynomial-size quantum proof~\cite{Watrous08}. In recent years there have been significant
advances on the structure of QMA systems, where multiple \emph{unentangled} proofs
and possibly locally restricted measurements in the verification were considered~\cite{KMT03,
ABDFS08, Harrow-Montanaro10, BCY11, Brandao-Harrow12}. It has been proven that many
natural problems in quantum physics are characterized by QMA proof systems (see, e.g.,
\cite{QMA-physics, LCV07, Chailloux-Sattath11, GHMW13}).

To solve a problem, the verifier performs a quantum algorithm on the input $x\in\{0,1\}^n$
along with the quantum proofs. The algorithm then returns ``yes'' or ``no'' as the answer to the
instance $x$. This procedure of verification can be effectively described as a set
of two-outcome measurements $\{(M_x, \1-M_x)\}_x$ on the proofs. In the definition below,
a problem is formally identified with a ``language''.

\begin{definition}
  \label{def:QMA}
  A language $L$ is in $\QMA^{\M}(k)_{m,c,s}$ if there exists a polynomial-time implementable
  verification $\{(M_x, \1-M_x)\}_x$ with each measurement from the class $\M$ such that
   \begin{itemize}
           \vspace{-1.5mm}
     \item Completeness: If $x\in L$, there exist $k$ states as proofs $\omega_1,\ldots ,\omega_k$,
           each of size $m$ qubits, such that
           \vspace{-2mm}
           \[ \tr\left(M_x(\omega_1\otimes\ldots\otimes\omega_k)\right)\geq c . \]
           \vspace{-8mm}
     \item Soundness: If $x\notin L$, then for any $\omega_1,\ldots ,\omega_k$,
           \vspace{-2mm}
           \[ \tr\left(M_x(\omega_1\otimes\ldots\otimes\omega_k)\right)\leq s . \]
  \end{itemize}
\end{definition}

We are also interested in QMA systems with multiple symmetric proofs. $\SymQMA^{\M}(k)_{m,c,s}$
is defined in a similar way but here we replace independent proofs $\omega_1,\ldots ,\omega_k$
with identical ones $\omega^{\ox k}$ in both completeness and soundness parts. As a convention,
we set $\M$ to be $\ALL$ (the class of all measurements), $m=poly(n)$, $k=1$, $c=2/3$ and $s=1/3$
as defaults~\cite{note2}. We can now state our application of Theorem~\ref{thm:de-Finetti} to
these complexity classes.

\begin{corollary}
  \label{cor:QMA}
  We have
  \begin{equation}
    \label{eq:QMA-reduction}
    \QMA=\QMA^{\LOCCONE}(poly)=\SymQMA^{\LOCCONE}(poly).
  \end{equation}
  In particular,
  \begin{align}
    \label{eq:sym-1locc-QMA}
    &\SymQMA^{\LOCCONE}(k)_{m,c,s}\subseteq \QMA_{0.6m^2k^2\epsilon^{-2},c,s+\epsilon}, \\
    \label{eq:1locc-QMA}
    &\QMA^{\LOCCONE}(k)_{m,c,s}\subseteq \QMA_{0.6m^2k^4\epsilon^{-2},c,s+\epsilon}
  \end{align}
\end{corollary}

It has been proven in~\cite{BCY11} that $\QMA=\QMA^{\LOCCONE}(k)$ for constant $k$.
Our result generalizes this statement to a polynomial number of proofs. It is also a
generalization of the results in~\cite{Brandao-thesis,Brandao-Harrow12} which prove the
reduction of $\QMA^{\LO}(k)$ to $\QMA$ ($\LO$ denotes local measurements). On the other
hand, Ref.~\cite{Brandao-Harrow12} proved that, assuming ETH (exponential time
hypothesis for 3-$\SAT$)~\cite{IPZ98}, any multi-prover QMA protocol with symmetric
proofs and Bell verification for 3-$\SAT$, can not bring better than the square-root
reduction of~\cite{Chen-Drucker10} to the proof size. Eq.~(\ref{eq:sym-1locc-QMA})
implies that, this is still true even if \emph{adaptively} local verification (one-way
LOCC measurement) is permitted.

Arguably the biggest open question in the study of QMA proof systems is whether
$\QMA=\QMA(2)$ (note that Harrow and Montanaro have proved that $\QMA(2)=\QMA(k)$ for
any polynomial $k>2$~\cite{Harrow-Montanaro10}). On the one hand, there are natural
problems from quantum physics that are in $\QMA(2)$ but not obviously in
$\QMA$~\cite{LCV07, Chailloux-Sattath11, GHMW13}. On the other hand, Harrow and Montanaro
showed that if the first equality in Eq.~(\ref{eq:QMA-reduction}) holds for a kind of
separable measurements (even only for the case of two proofs), then $\QMA=\QMA(2)$.
Our result here, although does not touch this open question directly, is a step towards
a larger measurement class compared to~\cite{Brandao-Harrow12} and we hope it will
stimulate future progress in solving this open question.

\medskip\noindent
{\em Polynomial optimization over hyperspheres.}
Theorem~\ref{thm:de-Finetti} also gives some improved results on the usefulness of a
general relaxation method, called the Sum-of-Squares (SOS) hierarchy~\cite{Lasserre01, Parrilo00},
for polynomial optimization over hyperspheres (see, e.g.,~\cite{Brandao-Harrow12, BKS14}). The
relevance in physics is that pure states of a quantum system form exactly a hypersphere
and hence some computational problems in quantum physics are indeed to optimize a polynomial
over hyperspheres. See Appendix for the details.

\medskip\noindent
{\bf Discussions.}
The advantage of our method, inherited from~\cite{Brandao-Harrow12}, is that it
tells us more information than that of~\cite{BCY11, Li-Winter14} about the valid de
Finetti (separable) state that approximates the symmetric (extendible) state. As a result,
we obtain a huge improvement over~\cite{Brandao-Christandl11} on the particle-number
dependence, and we are able to strengthen the relation $\QMA=\QMA^{\LOCCONE}(k)$ from the
constant $k$ of~\cite{BCY11} to polynomial $k$. We hope that the de Finetti theorem
presented in this letter will find more applications in the future.

We ask whether Theorem~\ref{thm:de-Finetti} can be further improved, to work for
two-way LOCC or even separable measurements. This would accordingly give stronger
applications, and possibly, solve the $\QMA$ vs $\QMA(2)$ puzzle due to the result
of~\cite{Harrow-Montanaro10}. Another open question is, for a state supported on the
symmetric subspace (aka Bose-symmetric state), whether its reduced states have pure-state
approximations of the form  $\int\varphi^{\ox k}\,\di\mu(\varphi)$ with $\varphi$
\emph{pure} that are not worse than the mixed-state approximations given by our theorem.
We notice that this is indeed the case for the de Finetti theorem of~\cite{CKMR07} and a
similar statement holds for \cite{Renner07}. However, our method, as well as that
of~\cite{Brandao-Harrow12} seems to require that the state $\varphi$ must be generally mixed.

\medskip\noindent
{\bf Acknowledgements.} KL is supported by NSF Grant CCF-1110941 and CCF-1111382.
GS acknowledges NSF Grant CCF-1110941. We thank Charles Bennett, Fernando Brand\~{a}o,
Aram Harrow and John Smolin for interesting discussions, and the anonymous referees
for helping improve the manuscript.

\bigskip
\appendix
{\bf \Large Appendix} \\

\noindent
{\bf Inequivalence of $\| \cdot \|_{\LOCCONE}$ and $\| \cdot \|_{\LOCCONEP}$.}
Here we show the following: for all $d$ there are constant $C$ and $d\times d\times 2$ states
$\rho_{ABC}$ and $\sigma_{ABC}$ such that $\| \rho_{ABC}-\sigma_{ABC} \|_{\LOCCONE} = 2$ but
$\| \rho_{ABC}-\sigma_{ABC} \|_{\LOCCONEP}\leq C/\sqrt{d}$ .

To see this, notice that for states of the form $\rho_{ABC} = \rho_{AB}\otimes \proj{0}$ and
$\sigma_{ABC} = \sigma_{AB}\otimes \proj{0}$ we have
\begin{align*}
  &\| \rho_{ABC} - \sigma_{ABC} \|_{\LOCCONE}= \| \rho_{AB}-\sigma_{AB}\|_{\LOCCONE}, \\
  &\| \rho_{ABC} - \sigma_{ABC} \|_{\LOCCONEP} =\| \rho_{AB}-\sigma_{AB}\|_{\LO},
\end{align*}
where $\LO$ denotes the set of local measurements. We can then apply the existence of
bipartite states with $\| \rho_{AB}-\sigma_{AB}\|_{\LOCCONE} = 2$ and
$\| \rho_{AB}-\sigma_{AB}\|_{\LO} \leq C/\sqrt{d}$ as shown in Theorem 2.4 of~\cite{Aubrun-Lancien14}.

\bigskip

\begin{proof-of}[of Lemma~\ref{lemma:monogamy}]
Let $\rho^{AZ_1\ldots Z_m}:=\id^A\ox\cM_1\ox\ldots\ox\cM_m(\rho^{AB_1\ldots B_m})$, with
$\cM_\ell^{B_\ell\rar Z_\ell}$ being measurement operations. Due to the chain rule of mutual
information,
\begin{equation}\begin{split}
  \label{eq:chain-rule}
  I(A;Z_1\ldots Z_m)=I(A;&Z_1)+I(A;Z_2|Z_1)+\ldots  \\
                         &+I(A;Z_m|Z_1\ldots Z_{m-1}).
\end{split}\end{equation}
Now we fix a special choice of $\cM_\ell\text{'s}$. Let $\cM_1$ maximize $I(A;Z_1)$.
Then under this choice of $\cM_1$, we choose $\cM_2$ that maximizes $I(A;Z_2|Z_1)$. Repeat
this procedure and at last we pick $\cM_m$ that maximizes $I(A;Z_m|Z_1\ldots Z_{m-1})$ under
the previously fixed measurements $\cM_1, \ldots, \cM_{m-1}$. As a result, for each conditional
mutual information we have
\begin{equation}\begin{split}
   \label{eq:proof-monogamy-1}
        &I(A;Z_\ell|Z_1\ldots Z_{\ell-1}) \\
   \geq &\min_{\cQ'}\max_{\cP'}I(A;Y_\ell|X_\ell)_{id^A\ox\cP'\ox\cQ'(\rho^{AB_1\ldots B_\ell})}\,,
\end{split}\end{equation}
where $\cP': B_\ell\rar Y_\ell$ and $\cQ': B_1\ldots B_{\ell-1}\rar X_\ell$ are measurement
operations. Further relax the minimization to allow the measurement $\cQ'$ to be performed on
all the $B$ systems except for $B_\ell$. Then we can set $\ell$ to be $1$ without changing
the value due to the symmetry of the state. Thus Eq.~(\ref{eq:proof-monogamy-1}) is further
lower bounded by
\[
  \min_{\cQ}\max_{\cP}I(A;Y|X)_{id^A\ox\cP\ox\cQ(\rho^{AB_1\ldots B_m})}
\]
with measurement operations $\cP: B_1\rar Y$ and $\cQ: B_2\ldots B_m\rar X$. This, combined with
Eq.~(\ref{eq:chain-rule}) lets us conclude that
\[\begin{split}
  \log|A| &\geq I(A;Z_1\ldots Z_m) \\
          &\geq m \min_{\cQ}\max_{\cP}I(A;Y|X)_{id^A\ox\cP\ox\cQ(\rho^{AB_1\ldots B_m})}\,,
\end{split}\]
and we are done.
\end{proof-of}

\bigskip

\begin{proof-of}[of Corollary~\ref{cor:det-entanglement}]
We first prove this result with the error measured by the fully one-way LOCC norm. Then
we explain that a slight adaptation works for the case of relative entropy.

Let $\ell\geq k$ be an integer. Introduce quantum systems $\bar{A}_i:=A_1^iA_2^i\ldots
A_k^i$ with $A_j^i\cong A_j$, for all $0<i\leq\ell$ and $0<j\leq k$. We search for an
state $\tilde{\rho}_{\bar{A}_1\ldots\bar{A}_\ell}$ such that $\tilde{\rho}_{\bar{A}_1
\ldots\bar{A}_\ell}$ is permutation-invariant and $\tilde{\rho}_{A_1^1A_2^2\ldots A_k^k}
=\rho_{A_1A_2\ldots A_k}$. If such a state exists, we feed back ``separable''. Otherwise
we conclude that it is entangled. For our purpose we set $\ell=2(k-1)^2\epsilon^{-2}\sum_i
\ln|A_i|+k$.

This search can be done using semidefinite programming in time polynomial of the total
dimension of $\tilde{\rho}_{\bar{A}_1\ldots\bar{A}_\ell}$, which coincides with
the claimed result.

To analyze the correctness, first assume that such an extension exists. Then we apply
Theorem~\ref{thm:de-Finetti} to see that there is certain probability measure $\mu$ such
that
\[\left\|\tilde{\rho}_{\bar{A}_1\ldots\bar{A}_k}-\int\sigma^{\ox k}\,\di\mu(\sigma)
            \right\|_{\LOCCONE}\leq\epsilon.\]
By definition, if we restrict the measurement to be performed only on systems
$A_1^1, A_2^2,\ldots , A_k^k$ then the above inequality implies that there exists a separable
state $\sigma_{A_1\ldots A_k}$ such that $\left\|\rho_{A_1\ldots A_k}-\sigma_{A_1\ldots
A_k}\right\|_\LOCCONE\leq\epsilon$. So if $\rho_{A_1\ldots A_k}$ is $\epsilon$-away from
any separable state, the required extension can not exist. But if $\rho_{A_1\ldots A_k}$ is
separable, it is obvious that such an extension does exist. As a result in both cases the
above procedure works correctly.

The above argument works as well if the error is measured by the relative entropy. The
small modification needed is just to replace $\|\cdot\|_\LOCCONE$ by $D_\LOCCONE$ and
here we set $\ell=(k-1)^2\epsilon^{-1}\sum_i\log|A_i|+k$.
\end{proof-of}

\bigskip
\noindent
{\bf Closeness of extendible states to being separable.} We also show how an extendible
multipartite state is close to the set of separable states, under fully one-way LOCC
distinguishability. A state $\rho_{A_1\ldots A_k}$ is $\ell$-extendible if there is an
extension $\tilde{\rho}_{\bar{A}_1\ldots\bar{A}_\ell}$  with $\bar{A}_i:=A_1^iA_2^i\ldots
A_k^i$, such that for all $0<j\leq k$ the state $\tilde{\rho}_{\bar{A}_1\ldots\bar{A}_\ell}$
is invariant under any permutations over subsystems $A_j^1, A_j^2,\ldots, A_j^\ell$ and
for any $0<i_1,\ldots, i_k\leq\ell$ we have $\rho_{A_1\ldots A_k}=\tilde{\rho}_{A_1^{i_1}\ldots A_k^{i_k}}$.
Obviously $\tilde{\rho}_{\bar{A}_1\ldots\bar{A}_\ell}$ is permutation-invariant and
$\tilde{\rho}_{A_1^1A_2^2\ldots A_k^k}=\rho_{A_1A_2\ldots A_k}$. So similar to the
argument in the proof of Corollary 4, a use of Theorem 1 lets us obtain:
\begin{align*}
  E_r^{\LOCCONE}(\rho_{A_1\ldots A_k}) &\leq \frac{(k-1)^2\sum_i\log|A_i|}{\ell-k},\\
  \min_{\sigma\in\sep}\left\|\rho_{A_1\ldots A_k}-\sigma_{A_1\ldots A_k}\right\|_\LOCCONE
      &\leq \sqrt{\frac{2(k-1)^2\sum_i\ln|A_i|}{\ell-k}}
\end{align*}
holds for any $\ell$-extendible state $\rho_{A_1\ldots A_k}$. Here $\sep$ denotes the
set of all separable states.

\bigskip

\begin{proof-of}[of Corollary~\ref{cor:QMA}]
Restricting the verification to be performed on the first proof in the multi-prover
protocols, we see that
\begin{align}
  \label{eq:QMA-proof-1}
  &\QMA_{m,c,s}\subseteq\SymQMA^{\LOCCONE}(k)_{m,c,s}\, ,     \\
  \label{eq:QMA-proof-2}
  &\QMA_{m,c,s}\subseteq\QMA^{\LOCCONE}(k)_{m,c,s}\, .
\end{align}
By definition, Eq.~(\ref{eq:QMA-proof-1}) and Eq.~(\ref{eq:sym-1locc-QMA}) imply
$\SymQMA^\LOCCONE(poly)=\QMA$. Similarly, Eq.~(\ref{eq:QMA-proof-2}) and Eq.~(\ref{eq:1locc-QMA})
imply $\QMA^\LOCCONE(poly)=\QMA$. Note that we can use the amplification of $\QMA_{m,c,s}$
(see~\cite{Marriott-Watrous05}) to keep $c=2/3$ and $s=1/3$.

To prove Eq.~(\ref{eq:sym-1locc-QMA}), we show a way of simulating a $\SymQMA^{\LOCCONE}
(k)_{m,c,s}$ protocol in a single-proof $\QMA$ system. The prover provides the verifier
with a proof of size $0.6m^2k^2\epsilon^{-2}$, which consists of $\ell=0.6mk^2\epsilon^{-2}$
subsystems each of size $m$ qubits. Then the verifier makes a uniformly random
permutation over the subsystems and then performs the $\SymQMA^{\LOCCONE}(k)_{m,c,s}$
verification on the first $k$ (denoted as $A_1, A_2, \ldots, A_k$) of them. No matter what
the initial state $\rho_{\text{initial}}$ of the proof is, Theorem~\ref{thm:de-Finetti}
implies that the state on $A_1\ldots A_k$, $\rho_{A_1\ldots A_k}$, can be approximated
as $\left\|\rho_{A_1\ldots A_k}-\int\sigma^{\ox k}\,\di\mu(\sigma)\right\|_{\LOCCONE}\leq
2\epsilon$ with certain probability measure $\mu$. Let $\{(M_x, \1-M_x)\}_x$ be the one-way
LOCC measurements in the $\SymQMA^{\LOCCONE}(k)_{m,c,s}$ protocol for a language $L$. Then
the soundness constant $s^\prime$ in this simulation can be
\[
  \begin{split}
  s^\prime & =   \max_{x\notin L}\max_{\rho_{\text{initial}}}\tr M_x\rho_{A_1\ldots A_k}  \\
           &\leq \max_{x\notin L}\max_\mu \int\di\mu(\sigma)\,\tr M_x\sigma^{\ox k}+\epsilon\\
           &\leq s+\epsilon.
  \end{split}
\]
One the other hand, suppose the proof for an accepted instance in $\SymQMA^{\LOCCONE}(k)
_{m,c,s}$ is $\omega^{\ox k}$. Then in the simulation the state $\omega^{\ox \ell}$
gives the same probability of acceptance. So completeness does not change.

For Eq.~(\ref{eq:1locc-QMA}), we will prove
\[\QMA^{\LOCCONE}(k)_{m,c,s} \subseteq \SymQMA^{\LOCCONE}(k)_{km,c,s}\, . \]
This, together with Eq.~(\ref{eq:sym-1locc-QMA}), leads to Eq.~(\ref{eq:1locc-QMA}).
The argument is similar to that in~\cite{ABDFS08} (Lemma 38), where the same
relation with ``$\LOCCONE$'' replaced by ``$\ALL$'' was proved. The strategy is to
divide each proof in the $\SymQMA^{\LOCCONE}(k)_{km,c,s}$ system into $k$ subsystems
of $m$ qubits, then simulate the $\QMA^{\LOCCONE}(k)_{m,c,s}$ protocol on the
$i\text{th}$ subsystem from the $i\text{th}$ proof for all $i=1,\ldots , k$.
\end{proof-of}

\bigskip\noindent
{\bf Polynomial optimization over hyperspheres.}
An immediate consequence of Theorem~\ref{thm:de-Finetti} is that we
can enlarge in~\cite{Brandao-Harrow12} the class of polynomials, for which the
optimization over multiple hyperspheres admits efficient SOS approximation. We also
provide another class of polynomials whose optimization over the single hypersphere
has a similar feature, supplementing a result of~\cite{BKS14} on polynomials
with nonnegative coefficients.

We use a $d$-dimensional complex vector to encode $2d$ real variables.

\begin{corollary}
  \label{cor:P-Optim}
  For $1\leq i\leq k$, let $A_i$ and $A_i^\prime$ be identical quantum systems of
  dimension $d$. Let $\ket{\alpha_i}\in\cH_{A_i}$ and $\ket{\beta}^{\ox k} \in
  \cH_{A_1A_1^\prime}\ox\ldots\ox\cH_{A_kA_k^\prime}$ be complex vectors. Let
  $0\leq M\leq\1$ be a matrix on  $\cH_{A_1\ldots A_k}$ such that $\{M, \1-M\}\in
  \LOCCONE$. The two optimizations
  \begin{align*}
  (\cO1:)\ &\max\  \bra{\alpha_1}\ox\ldots\ox\bra{\alpha_k}M\ket{\alpha_1}\ox\ldots\ox\ket{\alpha_k} \\
  &\textup{subject to }\  \braket{\alpha_i}{\alpha_i}=1, i=1,\ldots,k                          \\
  (\cO2:)\ &\max\  \bra{\beta}^{\ox k}(\1 \ox M) \ket{\beta}^{\ox k} \textup{  subject to  }\braket{\beta}{\beta}=1
  \end{align*}
  can be solved to within additive error $\epsilon$ efficiently, via a hierarchy
  of SDP relaxations (SOS), respectively in time $\exp(O(\epsilon^{-2}k^4\log^2(d))$
  and $\exp(O(\epsilon^{-2}k^2\log^2(d))$.
\end{corollary}

The advantage of Corollary~\ref{cor:P-Optim} is that, for constant $\epsilon$ and $k$, the
runtime of these two optimizations is only quasi-polynomial of the number of variables, instead
of exponential time of exhaustive search.

\medskip
\begin{proof}
The analysis of $\cO1$ is the same as that of the polynomial-optimization problem considered
in~\cite{Brandao-Harrow12}. We only need to employ our Theorem~\ref{thm:de-Finetti} when the
de Finetti theorem is used.

It is easy to see that the maximum in $\cO2$ equals
\[\max_\sigma \tr M \sigma^{\ox k},\]
with $\sigma$ a normalized quantum state. This, in turn, is bounded as
\[\begin{split}
       &\max_{\rho_{A_1\ldots A_\ell}}\tr(M\ox\1)\rho_{A_1\ldots A_\ell}-\sqrt{\frac{(k-1)^2\ln d}{2(\ell-k)}} \\
  \leq &\ \max_\sigma\tr M \sigma^{\ox k} \\
  \leq &\max_{\rho_{A_1\ldots A_\ell}}\tr(M\ox\1)\rho_{A_1\ldots A_\ell},
\end{split}\]
where $\ell\geq k$ and the maximization in the first and last lines are over
\emph{permutation-invariant} state $\rho_{A_1\ldots A_\ell}$. Note that here the
first inequality follows from a direct application of Theorem~\ref{thm:de-Finetti},
and the second inequality is by restricting the maximization in the last line to
over $\ell$-fold states of the form $\sigma^{\otimes\ell}$.  So the problem $\cO2$
can be approximated to within additive error $\sqrt{\frac{(k-1)^2\ln d}{2(\ell-k)}}$,
by the lever-$\ell$ SDP hierarchy (SOS hierarchy)
\[\begin{split}
\textup{maximize }  \ &\tr(M\ox\1)\rho_{A_1\ldots A_\ell} \\
\textup{subject to }\ &\rho_{A_1\ldots A_\ell}\geq 0, \tr\rho_{A_1\ldots A_\ell}=1, \\
                      &\rho_{A_1\ldots A_\ell} \text{being permutation-invariant}.
\end{split}\]
This can be done in time $\exp(O(\ell\log d)$, namely, polynomial of the dimension
of $\rho_{A_1\ldots A_\ell}$. At last, to obtain the claimed result, we choose
$\ell=\frac{1}{2}\epsilon^{-2}(k-1)^2\ln d +k$.
\end{proof}


\begin{thebibliography}{99}
\bibitem{note0}
The exchange of two systems $A_i$ and $A_j$ causes a unitary transformation
$U_{ij}\ket{\phi_{A_i}}\ket{\varphi_{A_j}}=\ket{\varphi_{A_i}}\ket{\phi_{A_j}}$ on their
state. We say $\rho_{A_1\ldots A_n}$ is permutation-invariant if $U_{ij}\rho_{A_1\ldots A_n}
U_{ij}^\dagger=\rho_{A_1\ldots A_n}$ for any $0<i<j\leq n$.
\bibitem{note1}
In the present paper, $\ln$ and $\log$ are logarithms with base $e$ and 2, respectively.
\bibitem{note2}
Actually the values of $c$ and $s$ can be chosen arbitrarily as long as $c-s\geq 1/poly(n)$,
because this gap can be amplified to have exponentially small errors. See~\cite{ABDFS08,
Marriott-Watrous05} for the amplification of $\QMA^{\LOCCONE}(k)$ and $\QMA$. The
amplification of $\SymQMA^{\LOCCONE}(k)$ follows from Eq.~(\ref{eq:sym-1locc-QMA}) together
with the amplification of $\QMA$.
\end{thebibliography}
\end{document}